\documentclass[conference,9pt]{IEEEtran}

\IEEEoverridecommandlockouts
\usepackage{cite}
\usepackage{amsmath,amssymb,amsfonts}
\usepackage{algorithmic}
\usepackage{graphicx}
\usepackage{textcomp}
\usepackage{changes}
\usepackage{marvosym}

\usepackage{bbding}
\usepackage{xcolor}
\usepackage{amsmath}
\usepackage{amssymb}
\usepackage{amsthm}
    \newtheorem{theorem}{Theorem}
    \newtheorem{lemma}{Lemma}

\usepackage{dblfloatfix} 
\usepackage{amsfonts}
\usepackage{graphicx}
\usepackage{caption}
\usepackage{float}
\usepackage{graphicx}
\usepackage{caption}
\usepackage{multicol}
\usepackage{multirow}
\usepackage{diagbox}
\usepackage{subfigure}
\usepackage{algorithm}
\usepackage{algorithmic}
\usepackage{braket}
\usepackage{CJKutf8}
\usepackage{qcircuit}
\usepackage{hyperref}
    \hypersetup{
        colorlinks=true,
        linkcolor=blue,
        filecolor=magenta,      
        urlcolor=cyan,
    }
\usepackage[export]{adjustbox}
\usepackage{balance}
\usepackage{bbding}
\usepackage{pifont}
\usepackage{booktabs}
\usepackage{adjustbox} 
\usepackage{tikz}
\usepackage{changes}

\newcommand{\figref}[1]{Fig.~\ref{#1}}
\newcommand{\tabref}[1]{Tab.~\ref{#1}}
\newcommand{\eqnref}[1]{Eqn.~\ref{#1}}
\newcommand{\algref}[1]{Alg.~\ref{#1}}
\newcommand{\lemref}[1]{Lemma~\ref{#1}}
\newcommand{\secref}[1]{Section~\ref{#1}}
\newcommand{\theref}[1]{Theorem~\ref{#1}}

\usepackage{booktabs}
\def\BibTeX{{\rm B\kern-.05em{\sc i\kern-.025em b}\kern-.08em
    T\kern-.1667em\lower.7ex\hbox{E}\kern-.125emX}}
\begin{document}

\title{CNOT Oriented Synthesis for Small-Scale Boolean Functions Using Spatial Structures of Parallelotopes
}
\author{

\IEEEauthorblockN{Qiang Zheng\textsuperscript{1}, Yongzhen Xu\textsuperscript{2}, Jiaxi Zhang\textsuperscript{3,\Letter}, Zhaofeng Su\textsuperscript{1,\Letter}, Shenggen Zheng\textsuperscript{2,\Letter}}

\IEEEauthorblockA{\textit{\textsuperscript{1}School of Computer Science and Technology, University of Science and Technology of China;}}

\IEEEauthorblockA{\textit{\textsuperscript{2}Quantum Science Center of Guangdong-Hong Kong-Macao Greater Bay Area;}}

\IEEEauthorblockA{\textit{\textsuperscript{3}School of Computer Science, Peking University.}}

\IEEEauthorblockA{\textsuperscript{\Letter}Corresponding authors: zhangjiaxi@pku.edu.cn, zfsu@ustc.edu.cn, zhengshenggen@quantumsc.cn}

\vspace{-2em}
}

\maketitle

\begin{abstract}
Quantum computing has garnered significant interest for its potential to achieve exponential speedups over classical approaches. However, in the Noisy Intermediate-Scale Quantum (NISQ) era, quantum circuit scalability remains limited by gate fidelity and qubit counts, restricting physical implementations to small-scale circuits. While prior work has explored logic network structures for quantum circuit synthesis, these methods often neglect the spatial structure intrinsic to Boolean functions.
In this paper, we leverage this spatial structure, encoded by parallelotopes embedded in the hypercube defined by the Boolean function, to access a broader optimization space, enhancing synthesis efficiency and reducing circuit complexity. 
We propose the Spatial Structure-based Hypercube Reduction~(SSHR), a novel synthesis method tailored for small-scale Boolean  functions ($\leq 8$). SSHR extracts global spatial features to minimize the use of Multi-Control Toffoli (MCT) gates.
To further exploit spatial correlations, we introduce two variants: SSHR-H employs heuristic functions to accelerate synthesis runtime, while SSHR-I integrates an Integer Linear Programming (ILP) solver to maximize spatial structure utilization. Our approach outperforms existing techniques in small-scale circuit synthesis, achieving 56\% and 81\% reductions in CNOT gate counts compared to the Exclusive Sum-of-Products (ESOP) and Xor-And-Inverter Graph (XAG) methods, respectively.
\end{abstract}


\section{Introduction}\label{sec:introduction}
Quantum computing is gaining significant attention due to its potential for  exponential speedup in certain computational tasks. 
Many famous quantum algorithms have solved different problems due to their ingenious design, e.g., Grover's algorithm~\cite{grover1996fast} for database searching, Shor's algorithm~\cite{shor1999polynomial} for prime factorization, and  
Harrow-Hassidim-Lloyd~(HHL) algorithm for linear systems of equations~\cite{harrow2009quantum}.

The execution of quantum algorithms is realized through quantum circuits. 
In the NISQ era, reducing the complexity of circuits is crucial, as it helps us obtain higher experimental accuracy and execute larger-scale quantum algorithms.
Oracle, as a black-box function, often plays an important role in many famous quantum algorithms, such as Grover's algorithm~\cite{grover1996fast}, Simon's algorithm~\cite{simon1997power}. 
Therefore, the synthesis of Oracle is crucial, which determines the testing of quantum algorithms and attracts attention~\cite{bravyi20226}.

For a Boolean function Oracle \( f: \{ 0,1 \}^{n} \to \{ 0,1\} \), its quantum circuit \( O_f \) operates as described by the equation \( O_f \ket{x}\ket{y}\ket{0}^k=\ket{x}\ket{y\oplus f(x)}\ket{0}^k \) for \( x\in \{ 0,1 \}^{n} \) and \( y\in \{ 0,1\} \) with $k$ auxiliary qubits~\cite{montanaro2008quantum}. 
Since Multi-Control Toffoli (MCT) gates are a universal set of gates for classical logic and naturally satisfy reversibility, quantum circuits usually implement Boolean functions through MCT gates. 

Currently, two common methods exist for implementing quantum circuits for Boolean functions. 
Exclusive Sum of Products~(ESOP) is able to represent Boolean   functions in a more concise form by connecting multiple product terms (AND terms) through an Exclusive OR (XOR)~\cite{fazel2007esop}. CNOT gates in quantum circuits correspond directly to the XOR operation, while multi-control gates (e.g., Toffoli gates) enable efficient implementation of the product terms. ESOPs allow the result of each product term to be directly superimposed on the target bit through the accumulative property of XOR. However, since ESOP can utilize the global structure of Boolean functions, exact methods that try to find an ESOP form with a minimum number of product terms can hardly deal with large-scale Boolean variables effectively. Xor-And-Inverter Graph~(XAG) is the best quantum circuit at the current state-of-the-art synthesis technique, which is a logic network based on $\{\oplus ,\wedge,1\}$ \cite{meuli2022xor}. XAG utilizes auxiliary qubits to store intermediate results, reducing the number of Toffoli gates in quantum circuit using local optimization.
As a result, XAG performs well in large-scale circuits. In addition to this, there are many other methods for constructing special Boolean functions, such as autosymmetric~\cite{BernasconiBCCF23}, symmetric functions~\cite{zi2024shallow}, and others~\cite{patel2008optimal,chattopadhyay2014constructive}. 

Although we already have many quantum algorithms, it is very difficult to execute large-scale quantum circuits on real quantum computers due to the limitations from the number of qubits, fidelity, coherence time, and so on. For example, Fallek et al. verified the Bernstein–Vazirani algorithm with 3 qubits \cite{fallek2016transport} and Debnath et al. verified the Deutsch–Jozsa algorithm with 5 qubits \cite{debnath2016demonstration}. It can be seen that our practical applications of quantum computing are still at a small scale for quantum query algorithms.

This naturally raises the question: can we combine the strengths of both ESOP and XAG to develop a more efficient approach on small-scale circuits? We give an affirmative answer and take the first step in this direction. Specifically, we propose a method for implementing Boolean functions using quantum circuits, which outperforms ESOP and XAG. 
We extract the spatial structure of the coordinates within the hypercube represented by the Boolean function and synthesize a quantum circuit based on the structure of the parallelotopes.
Then we design Spatial
Structure-based Hypercube Reduction~(SSHR) that can store intermediate results. Since we utilize the structure of the Boolean function in the hypercube, our algorithm has access to the global structure and thus performs well on small-scale circuits.

We have listed the relevant properties of SSHR, ESOP, and XAG in \tabref{compared}.
ESOP offers a straightforward logical representation; however, its inability to store intermediate logic limits its effectiveness in quantum circuit synthesis. XAG has limited optimization on small-scale circuits due to its heuristic strategy, which is more applicable to large-scale circuits.

\begin{table}[htbp]
\caption{Comparing SSHR to ESOP and XAG. We compared them in terms of ancillary qubits, whether intermediate results are utilized, whether global structure is used, variable scale, and the type of quantum circuit gate.}
\label{compared}
\centering
\resizebox{\linewidth}{!}{
\begin{tabular}{|c|c|c|c|}
\hline
                     & ESOP                         & XAG                          & SSHR                         \\ \hline
Ancillary qubits     & \textcolor{red}{\ding{56}}   & \textcolor{green}{\ding{52}} & \textcolor{red}{\ding{56}}   \\ \hline
Intermediate results & \textcolor{red}{\ding{56}}   & \textcolor{green}{\ding{52}} & \textcolor{green}{\ding{52}} \\ \hline
Global structure   & \textcolor{green}{\ding{52}} & \textcolor{red}{\ding{56}}   & \textcolor{green}{\ding{52}} \\ \hline
Variable scale               & small                        & large                        & small                        \\ \hline
Gate type            & MCTs                         & X,CNOT,Toffoli               & MCTs                         \\ \hline
\end{tabular}
}
\end{table}

In this paper, we propose a quantum Boolean   function synthesis method~SSHR according spatial structure.
SSHR extract the global spatial structure of Boolean   functions and does not need to use auxiliary qubits to store intermediate results. In addition, we statute the synthesis problem of quantum circuits as a special set cover problem and solve it via Integer Linear Programming~(ILP) solver. This maximizes the use of the spatial structure of the Boolean   functional correlation prallelotopes, leading to better quantum circuits. Our main contributions are as follows:

\begin{itemize}
    \item We designed a Spatial
Structure-based Hypercube Reduction~(SSHR) that utilizes the global spatial structure of parallelotopes in Boolean   function and stores it without the use of auxiliary qubits. 
    
    \item We propose two algorithms, SSHR-H and SSHR-I, to effectively exploit the spatial structure. SSHR-H is a heuristic algorithm designed for faster circuit generation. SSHR-I formulates the quantum synthesis task as a Weighted Parity Set Covering Problem (WP-SCP), optimizing spatial structure utilization via ILP solver. 

    \item The experimental results demonstrate that SSHR-H reduces CNOT gates by 74.69\% vs. XAG (6-bit) and scales better than ESOP. SSHR-I reduces CNOT gates by 56\% (vs. ESOP) and 81\% (vs. XAG) for 5-bit, while both slash single-qubit gates and auxiliary qubits, proving high efficiency.
\end{itemize}
\section{Preliminaries}\label{sec:preliminaries}
This section begins with an introduction to Boolean functions and NPN equivalence classes, followed by an overview of quantum Oracle synthesis.

\subsection{Boolean function and NPN equivalence}
A Boolean function $f:\{ 0,1 \}^{n} \to \{ 0,1\}$ can be represented by a truth table, which is a bit string $f_{2^n-1}\cdots f_1 f_0$. Here, $f_x$ denotes the output of $f$ for a given input $x$. We refer to the hexadecimal value of the bit string  $f_{2^n-1}\cdots f_1 f_0$  as the ID of $f$. An input $x$ is considered a minterm if $f(x)=1$ \cite{brayton1990multilevel}. For the sake of clarity, we refer to on-set as those $x$ values that yield $f(x)=1$ and off-set as those $x$ values that yield $f(x)=0$. The satisfaction count of $f$ refers to the number of minterms for which $f$ evaluates to 1, denoted as $|f|$. 

An \( n \)-bit Boolean function can be conceptualized as an \( n \)-dimensional hypercube. The vertices of the cube correspond to different coordinates, with those in the on-set labeled red and those in the off-set labeled white. In this context, a parallelotope, as defined in reference \cite{gover2010determinants}, plays a crucial role in this paper. 
A parallelogram represents a 2-dimensional parallelotope, whereas a parallelepiped represents a 3-dimensional parallelotope. 
In \( n \)-dimensional space, the diagonals of parallelotope bodies intersect at a vertex and are bisected by that vertex. \figref{fig:hypercube1} illustrates a 2-dimensional parallelotope in 3-dimensional space, while \figref{fig:hypercube2} depicts a general 3-dimensional hyperparallelepiped. Similarly, we can extend this concept to define an \( m \)-dimensional parallelotope in \( n \)-dimension space.

\begin{figure}[htbp]
\centering
    \subfigure[A 2-dimensional parallelotope in a 3-dimension space.]{ 
        \begin{minipage}[t]{0.35\linewidth}
        \centering
        \includegraphics[width=\textwidth]{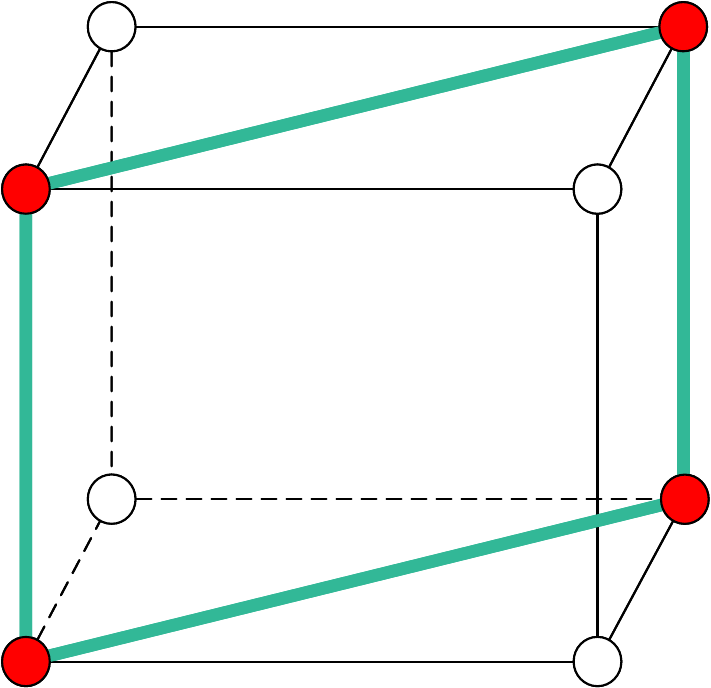}
        \end{minipage}
        \label{fig:hypercube1}
    }
    \hfill
    \subfigure[A 3-dimensional parallelotope in a 4-dimension space.]{ 
        \begin{minipage}[t]{0.35\linewidth}
        \centering
        \includegraphics[width=\textwidth]{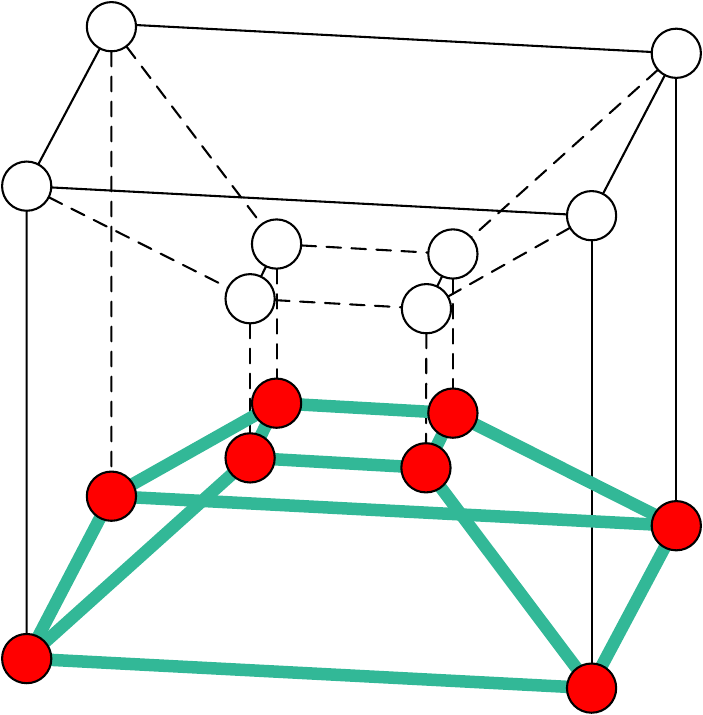}
        \end{minipage}
        \label{fig:hypercube2}
    }
\vspace{-1em}
    \caption{Examples of low-dimensional parallelotope.}
    \label{fig:example_Boolean _equivalence}
\end{figure}

Let $N=2^n$, there are $2^N$ $n$-bit Boolean functions in total. This bi-exponential growth greatly increases the number of Boolean functions, making analysis on a case-by-case basis impossible. For instance, there are 256 functions when $n=3$, and $2^{2^4}=65536$ when $n=4$. To simplify the analysis, we adopt the concept of NPN equivalence.  NPN equivalence, commonly used in the classification of Boolean functions \cite{zhang2023fast,zhang2022heuristic}, involves transformations of input negation, input permutation, and output negation. It finds significant applications in logic synthesis, technical mapping, and verification.

\figref{fig:npn} illustrates equivalence and non-equivalence in hypercubes \cite{zhang2023rethinking}. Although there are 256 3-bit functions, it is evident that there are only 14 functions in terms of NPN equivalence. For \( n=4 \), with a total of \( 2^{2^4}=65536 \) functions, NPN classification reduces them to 222 functions. NPN is also a crucial concept in quantum query complexity; if \( f \) and \( g \) are NPN equivalent, their query complexity remains the same \cite{ambainis2014exact}.

\begin{figure}[tbp]
\centering
    \subfigure[3-Majority $f_1$]{ 
        \begin{minipage}[t]{0.13\textwidth}
        \centering
        \includegraphics[width=\textwidth]{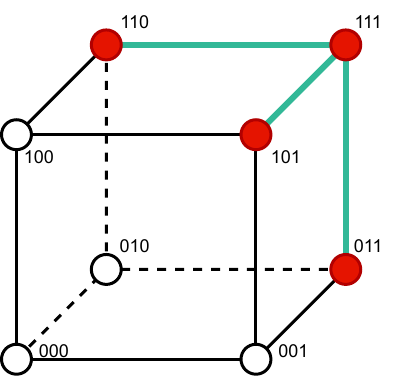}
        \end{minipage}
        \label{fig:npn1}
    }
    \hfill
    \subfigure[$f_2$]{ 
        \begin{minipage}[t]{0.13\textwidth}
        \centering
        \includegraphics[width=\textwidth]{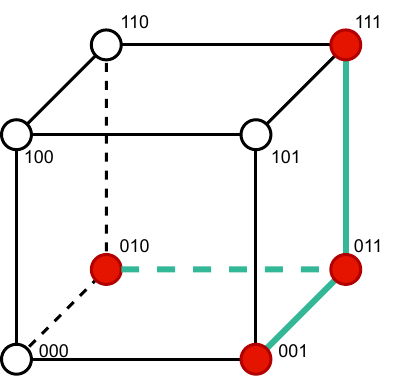}
        \end{minipage}
        \label{fig:npn2}
    }
    \hfill
    \subfigure[$f_3$]{ 
        \begin{minipage}[t]{0.13\textwidth}
        \centering
        \includegraphics[width=\textwidth]{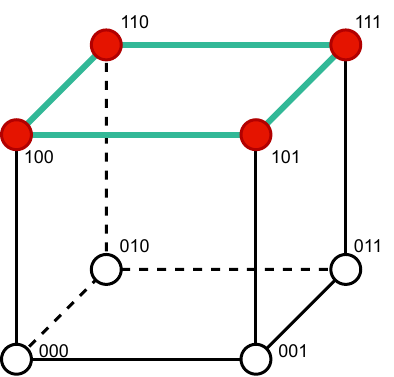}
        \end{minipage}
        \label{fig:npn3}
    }
\vspace{-1em}
    \caption{Hypercubes of three 3-variable Boolean functions. 3-majority logic $f_1$ and $f_2$ are NPN equivalent, and their induced subgraphs are isomorphic. $f_2$ and $f_3$ are not NPN equivalent, and their induced subgraphs are non-isomorphic.}
    \vspace{-\baselineskip}
    \label{fig:npn}
\end{figure}

\subsection{Quantum gate and circuit}
Quantum circuits offer a graphical framework for executing quantum computing operations through a sequence of quantum gates and measurements~\cite{nielsen2010quantum}. 
This paper employs several fundamental quantum gates, including the X-gate, H-gate, and T-gate, alongside the CNOT-gate, Toffoli-gate, and MCT-gate~\cite{miller2011elementary}. 
The matrix representations and circuit symbols of these gates are detailed in~\cite{nielsen2010quantum}.

A Quantum Oracle is typically treated as a black box, with the developers of quantum algorithms concentrating more on its functionality than on its concrete implementation.
A generalized example is illustrated in \figref{fig:oracle}.
For a Boolean function, $x$ denotes the state of the $n$-qubit of the input, and $y$ denotes the corresponding output, which includes 
$k$ auxiliary bits.
The input states and auxiliary bits remain unchanged after the Oracle, and the output value of the Boolean function is stored in $y$.
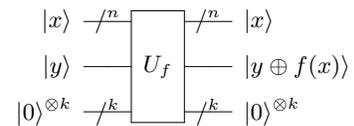
\begin{figure}[thb]
    \centering
    \[
    \Qcircuit @C=1em @R=1em {
        & \lstick{\ket{x}} & \qw  /^{n}& \multigate{2}{U_f} & \qw /^{n} & \rstick{\ket{x}} \qw \\
        & \lstick{\ket{y}} & \qw & \ghost{U_f} & \qw & \rstick{\ket{y \oplus f(x)}} \qw \\
        & \lstick{\ket{0}^{\otimes k}} & \qw /^{k}& \ghost{U_f} & \qw /^{k}& \rstick{\ket{0}^{\otimes k} \; } \qw \\
    }
    \]
    \vspace{-1em}
    \caption{Circuit demonstration of a Quantum Oracle where the input is $\ket{x}$ and the output result is stored on $\ket{y}$, along with $k$ auxiliary qubits initialized to $\ket{0}$.}
    \label{fig:oracle}
\end{figure}

Due to the current limitations in the physical implementation of quantum devices, most quantum computers achieve universality through the construction of a universal set of gates~\cite{divincenzo1995two}. 
Operations involving three or more qubits are not directly implemented on physical quantum devices; instead, they must be decomposed during compilation into a sequence of gates from the existing gate set~\cite{barenco1995elementary}.
This set of gates typically comprises single-qubit and two-qubit gates, which together provide the generality required for universal quantum computation~\cite{divincenzo1995two,plesch2011quantum}.
To facilitate comparison with other logic synthesis algorithms and ensure consistency in the metrics, we compute the cost of the $k$-MCT gate according to~\cite{BernasconiBCCF23}, which has been shown in \tabref{tab:decompose}. 
Additionally, the fidelity of two-qubit gate~(usually CNOT gate) is often an order of magnitude lower than that of single-qubit gates. 
Therefore, the number of CNOT gates plays a crucial role in determining the fidelity of the final implemented circuit.
Moreover, quantum circuits require a significant number of T-gates for fault-tolerant computation~\cite{gottesman1997stabilizer}. 
As T-gates are non-Clifford gates that demand higher precision control, their count becomes a critical factor to consider.  
Accordingly, SSHR addresses this limitation by employing a set-valued objective function to optimize for the number of T gates~(T-counts) and other relevant factors. 
This will be demonstrated in the experimental section, where we optimize the number of CNOT gates and T-counts separately.


\begin{table}[htb]
\centering
\vspace{-0.5em}
\caption{Cost of $k$-MCT gates in number of T, H and CNOT gates.}
\label{tab:decompose}
\begin{tabular}{|c|ccc|c|}
\hline
k                & T    & H     & CNOT & Ancillary qubits \\ \hline
2                & 7    & 2     & 6    & 0                \\
3                & 16   & 6     & 14   & 1                \\
$\geq4$ & 8k-8 & 8k-12 & 4k-6 & $\lceil \frac{k-2}{2} \rceil$           \\ \hline
\end{tabular}
\end{table}

\section{Spatial structure}\label{spatial}
In this section, we first define the spatial structure of Boolean functions. Building upon the concept of NPN equivalence, we then present a method for generating quantum circuit blocks using this spatial structure, along with a concrete algorithm to illustrate the approach.

For clarity and ease of use in the following sections, we first introduce several foundational concepts and operations.
\begin{itemize}
\item $INIT\_CIRCUIT(n)$: Initializes a circuit with $n$ qubits.
\item $INIT\_BLOCK\_LIST$: Initializes an empty list to store circuit blocks.
\item $ADD\_CNOT(i,j)$: Adds a CNOT gate with control position $i$ and target position $j$.
\item $ADD\_X(i)$: Adds an X-gate at position $i$.
\item $ADD\_MCT (i,\ldots, j,k)$: Adds an MCT gate with control positions $i, \ldots, j$, and $k$ as the target position.
\item $ILP\_SOLVER(A,S)$: Use the ILP solver to solve for the solution that covers the minimum term A in the set S.
\end{itemize}

The core of our algorithm lies in utilizing the structure of parallelotope forms embedded in the vector space. Let $V$ be a vector space over $\mathbb{R}^n$, and let $\alpha_1,\alpha_2,\linebreak[0] \ldots,\linebreak[0] \alpha_k$ be $k$ linearly independent vectors in $V$. A parallelotpe is then defined as the Minkowski sum of line segments, formally expressed as: 
\begin{equation}\label{eq:p}
    \{ v=\sum t_i \alpha_i \mid 0\leq t_i \leq 1\}.
\end{equation}

These vectors represent the generating vertices of the parallelotope $\mathcal{P}(\alpha_1,\alpha_2,\ldots,\alpha_k)$ \cite{gover2010determinants,sikiric2014sum}. The number of generating vectors determines the dimension of the parallelotope, while the dimension of each vector defines the dimension of the ambient space.

The preceding definition is situated in Euclidean space. We now restrict the space \(V\) to the hypercube over $n$-variable Boolean functions, where the coordinates take values in \(\{0,1\}\) and $\alpha_i\in\{0,1\}^n$. This imposes an additional constraint on the parallelotope: its basis vectors must satisfy the condition specified in \lemref{lemma1}, which serves as the foundation for the subsequent proof procedure.

\begin{lemma}\label{lemma1}
In a \(n\)-dimensional hypercube, the basis vectors that constitute a \( k \)-dimensional parallelotope must satisfy the conditions: 
\begin{equation}
\alpha_i \cdot \alpha_j = 0 \quad \forall 1 \leq i,j \leq k, \ i \neq j.
\end{equation}
\end{lemma}
\begin{proof}
Suppose that in a \(n\)-dimensional hypercube, there exists a \(k\)-dimensional parallelotope that satisfies the following condition:
\[
\alpha_i \cdot \alpha_j \neq 0 \quad \exists 1 \leq i,j \leq k, \ i \neq j. \]

Without loss of generality, let $\alpha_1=(a_1,a_2,\cdots,a_n)$ and $\alpha_2 = (b_1,b_2,\cdots,b_n)$, with $\alpha_1 \cdot \alpha_2 \neq 0$. Then there must exist some index \(i\) such that $a_i\cdot b_i \neq 0$, which implies $a_i+b_i > 1$. As a result, certain vertex coordinates generated from \eqnref{eq:p} will lie outside the \(n\)-dimensional hypercube, leading to a contradiction.
\end{proof}

According to \lemref{lemma1}, the basis vectors of the parallelotope must satisfy the following property: for any fixed coordinate index, no two vectors have the value 1 at that position simultaneously.

\begin{theorem}\label{theorem1}
For a Boolean function $f:\{ 0,1 \}^{n} \to \{ 0,1\}$, if its on-set contains $2^m$ minterms that form an m-dimensional parallelotope, then the corresponding oracle can be implemented by a quantum circuit using one $(n-m)$-MCT gate.
\end{theorem}

\begin{proof}
To simplify the proof process, we express the basis vectors of the parallelotope in a more straightforward and intuitive form using the NPN concept. Without loss of generality, we select \(m\) vectors as follows:
    \begin{align}
        \alpha_1 &=
                        ( \overbrace{ 1, 1, \ldots , 1}^{k_1}, \overbrace{0, 0 , \ldots , 0 }^{n-s_1} ),
                    \\
        \alpha_2 &=
                        ( \overbrace{ 0, 0, \ldots , 0}^{s_1} ,\overbrace{ 1, 1, \ldots , 1}^{k_2}, \overbrace{0, 0 , \ldots , 0 }^{n-s_2}) ,
                    \\
        \vdots \nonumber \\
        \alpha_m &=
                        ( \overbrace{ 0, 0, \ldots , 0}^{s_{m-1}},\overbrace{ 1, 1, \ldots , 1}^{k_m}, \overbrace{0, 0 , \ldots , 0 }^{n-s_m}) ,
    \end{align}
    where $s_j := \sum_{i=1}^j k_i$.

Clearly, by setting the coordinates of the first vertex as 
$(\overbrace{0,0,\ldots,0}^n)$ and applying the Minkowski sum defined in \eqnref{eq:p}, we can determine the coordinates of the remaining vertices of the parallelotope.

Since the coordinates of the other vertices $(a_0,a_1,\ldots,a_{n-1})$ are obtained through addition, these vertices must satisfy the following properties:

    \begin{enumerate}
        \item The first $k_1$ values are identical, i.e.
        \begin{equation}
            a_0\oplus a_1=a_0\oplus a_2 = \cdots = a_0\oplus a_{s1-1}.
        \end{equation}
        \item The middle $k_2$ values are identical, i.e.
        \begin{equation}
            a_{k1}\oplus a_{k1+1}=a_{k1}\oplus a_{k1+2} = \cdots = a_0\oplus a_{s_2-1}.
        \end{equation}
        \item Similar for $k_3$ to $k_{m-1}$; 
        \item The last $k_m$ values are identical, i.e.
        \begin{equation}
            \begin{aligned}
                 a_{s_{m-1}} \oplus a_{s_{m-1}+1}  
                &= a_{s_{m-1}}\oplus a_{s_{m-1}+2}  \\&=\cdots
                = a_{s_{m-1}}\oplus a_{s_{m}-1}.
            \end{aligned}
        \end{equation}
    \end{enumerate}

    Based on the aforementioned properties, we can use \algref{algorithm:1} to synthesize quantum circuits. By applying \algref{algorithm:1}, synthesizing the $2^m$ minterms, which form an $m$-dimensional parallelotope in an $n$-dimensional space, requires at most one ($n-m$)-MCT gate.

    \begin{algorithm}[htbp]
    \caption{Synthesizing $2^m$ minterms in on-set which form an $m$-dimensional parallelotope in $n$-dimensional space.}
    \label{algorithm:1}
    \begin{algorithmic}[1]
        \REQUIRE The input of bases that meet the above conditions of the m-dimensional parallelotope.
        \ENSURE The circuit block is synthesized by the input.
        
        \STATE $circuit \gets INIT\_CIRCUIT(n+1)$
        \STATE $st \gets 0$
        \STATE $list\gets \text{empty list}$
        \FOR{$j = 1$ to $m$}
            \FOR{$i = 1$ to $k_j$}
                \STATE $ADD\_CNOT(st,st+i)$
                \STATE Push $st + i$ to the back of $list$
            \ENDFOR
            \STATE $st\gets st+k_j$
        \ENDFOR
        \FOR{$i \in list$}
            \STATE $ADD\_X(i)$
        \ENDFOR
        \STATE $ADD\_MCT(list,output)$
        \FOR{$i \in list$}
            \STATE $ADD\_X(i)$
        \ENDFOR
        \RETURN $circuit$
    \end{algorithmic}
    \end{algorithm}
    
    In \algref{algorithm:1}, the first ``for" loop of the double ``for" loop is used to traverse \(m\) dimensions of the parallelotope, and the second ``for" loop is used to indicate the $k_j$ values are identical where $k_j$ is determined by the previously described and to add MCT gates to the circuit and save the control positions in list. Subsequently,  we add X gates to implement a 0-controlled signal. The next step in the circuit is to add an MCT gate based on the control qubits saved in the list, which serves to ensure that the above conditions are satisfied simultaneously and revert the 0-controlled signal at last.
\end{proof}

\begin{figure}[tbp]
    \centering
    \[
\scalebox{1.0}{
\Qcircuit @C=1.0em @R=0.2em @!R { \\
\nghost{{q}_{0} :  } & \lstick{{q}_{0} :  } & \ctrl{1} & \ctrl{2} & \qw  &\ctrl{3} & \qw    & \qw      & \ctrl{3}          & \qw & \ctrl{2} & \ctrl{1} & \qw & \qw\\
\nghost{{q}_{1} :  } & \lstick{{q}_{1} :  } & \targ    & \qw      &  \qw &\qw      & \qw    & \ctrlo{1}               & \qw & \qw &  \qw & \targ & \qw & \qw\\
\nghost{{q}_{2} :  } & \lstick{{q}_{2} :  } & \qw      & \targ    &  \qw &\qw      & \qw    & \ctrlo{1}                & \qw & \qw &  \targ & \qw & \qw & \qw\\
 &  &   &    &  \raisebox{0.5em}{\vdots} & \qwx[1]   & \raisebox{0.5em}{\vdots} & \qwx[1]     & \qwx[1]       & \raisebox{0.5em}{\vdots} &   &   & &   \\
\nghost{{q}_{3} :  } & \lstick{{q}_{s_1-1} :  } & \qw & \qw & \qw & \targ & \qw & \ctrlo{2} & \targ  &  \qw & \qw & \qw & \qw & \qw\\
\nghost{{q}_{4} :  } & \lstick{{q}_{s_1} :  } & \ctrl{1} & \ctrl{2} &  \qw &\ctrl{3} & \qw  & \qw & \ctrl{3}& \qw &  \ctrl{2} & \ctrl{1} & \qw & \qw\\
\nghost{{q}_{5} :  } & \lstick{{q}_{s_1+1} :  } & \targ & \qw &  \qw &\qw & \qw & \ctrlo{1} & \qw  & \qw &  \qw & \targ & \qw & \qw\\
\nghost{{q}_{6} :  } & \lstick{{q}_{s_1+2} :  } & \qw & \targ &  \qw &\qw & \qw & \ctrlo{1}  & \qw &  \qw & \targ & \qw & \qw & \qw\\
 &   &   &    & \raisebox{0.5em}{\vdots} & \qwx[1]   & \raisebox{0.5em}{\vdots} & \qwx[1]     &\qwx[1]  &  \raisebox{0.5em}{\vdots} &  &   & &   \\
\nghost{{q}_{7} :  } & \lstick{{q}_{s_2-1} :  } & \qw & \qw & \qw & \targ & \qw & \ctrlo{2} & \targ  & \qw & \qw & \qw & \qw & \qw\\
&   &\text{\scriptsize{other cases}}  &    &   \raisebox{0.5em}{\vdots} & &  \raisebox{0.5em}{\vdots} & \qwx[1]     &  & \raisebox{0.5em}{\vdots} &  &   & &   \\
\nghost{{q}_{8} :  } & \lstick{{q}_{s_2} :  } & \ctrl{1} & \ctrl{2} & \qw & \ctrl{3} &\qw & \qw\qwx[1] & \ctrl{3} &  \qw & \ctrl{2} & \ctrl{1} & \qw & \qw\\
\nghost{{q}_{9} :  } & \lstick{{q}_{s_2+1} :  } & \targ & \qw &  \qw &\qw & \qw & \ctrlo{1} &  \qw & \qw &\qw & \targ & \qw & \qw\\
\nghost{{q}_{10} :  } & \lstick{{q}_{s_2+2} :  } & \qw & \targ & \qw & \qw &\qw & \ctrlo{1}  & \qw &  \qw &\targ & \qw & \qw & \qw\\
     &   &   &    &   \raisebox{0.5em}{\vdots} & \qwx[1]  &\raisebox{0.5em}{\vdots} & \qwx[1]     &\qwx[1] & \raisebox{0.5em}{\vdots} &  &  &   & &   \\
\nghost{{q}_{11} :  } & \lstick{{q}_{s_m-1} :  } & \qw & \qw &  \qw &\targ &\qw & \ctrlo{1} & \targ & \qw &  \qw & \qw & \qw & \qw\\
 \text{\scriptsize{common~control~qubits}}   &   &   &    &    &    & \raisebox{0.5em}{\vdots}& \qwx[1]     & &  &  &   & &   \\
\nghost{{q}_{12} :  } & \lstick{{q}_{out} :  } & \qw & \qw &  \qw &\qw & \qw& \targ & \qw &  \qw & \qw & \qw & \qw & \qw \\
\\ }}
\]
\vspace{-1.5em}
    \caption{An example of \algref{algorithm:1}.}
    \label{fig:circuit}
\vspace{-1em}
\end{figure}
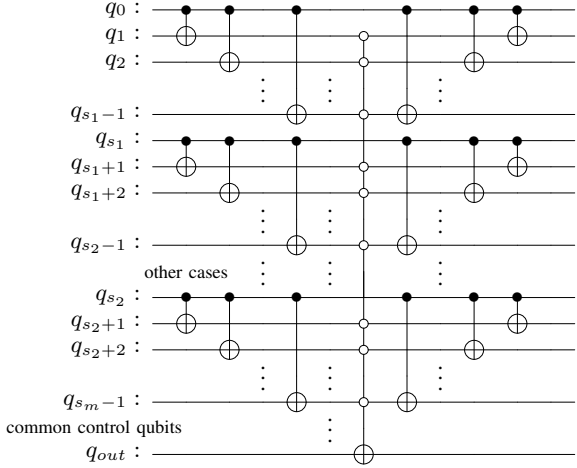

This structure of a parallelotope obtained from the basis vectors is global, since it depends on every coordinate in the vector. An example of \algref{algorithm:1} is shown in \figref{fig:circuit}. 
The CNOT gates and control signals on qubits \(q_0\) to \(q_{s_1-1}\) are based on the first basis vector obtained. Similarly, the CNOT gates and control signals on qubits \(q_{s_1}\) to \(q_{s_2-1}\) are based on the second basis vector. Additionally, control qubits corresponding to the common parts of these vectors are located on the qubits \(q_{s_m}-q_{n-1}\).

The coordinates of these points satisfy the conditions of the basis vectors. Using \algref{algorithm:1}, we can derive the corresponding circuit block, as shown in \figref{fig:circuit}.

According to \theref{theorem1}, we can find the following benefits of spatial structure like parallelotope in hypercube:
\begin{enumerate}
    \item The intermediate results are stored directly on the input qubits through the action of CNOT gates and are recovered after use, which leads to a stronger representation without the application of auxiliary qubits.
    \item This extraction of global structure is independent of order relations, ensuring that the final generated circuits are equivalent, which enhances layer exchangeability. 
\end{enumerate}
\section{Algorithm}
In this section, we design two algorithms, SSHR-H and SSHR-I, to utilize the spatial structure described above. The first algorithm is based on a greedy approach, offering a short runtime while ensuring better results. The second algorithm generalizes to the Weighted Parity Set Covering Problem~(WP-SCP) and employs an ILP solver, maximizing the use of spatial structure, while the objective function provides greater flexibility.

\subsection{SSHR-H}
In this section, we design a heuristic algorithm to efficiently select parallelotopes and maximize the utilization of spatial structure.

According to \tabref{tab:decompose}, we observe that the cost of a $k$-MCT gate increases with the number of control qubits $k$. At the same time, as stated in \theref{theorem1}, an increase in $k$ reduces the number of represented minterms, leading to a decrease in the representational capacity of the corresponding parallelotope.
This provides us with design insights for heuristic algorithms. We can argue that selecting a higher-dimensional parallelotope can effectively reduce the use of $k$-MCT gates, thereby lowering the overall circuit cost. Based on the spatial structure in $S$, we observe that when a parallelotope is selected, the corresponding minterms in its truth table result in a flip of the target qubit. Therefore, when we detect that the number of points in the minterms satisfies a certain ratio $R$ relative to the number of minterms in the intersection of a given parallelotope, we select this parallelotope and update the remaining set of minterms.

Based on the above idea, we propose a heuristic algorithm, SSHR-H, whose detailed procedure is described in \algref{alg:H}. Lines 1–2 initialize the quantum circuit and construct the initial set $A$, which contains all the minterms in the on-set. Line 3 computes the set $S$  of all candidate parallelotopes associated with the minterms in $A$. Lines 4–12 constitute the core of the algorithm. The goal is to iteratively cover all minterms in $A$, and the algorithm terminates once $A=\emptyset$ . Within each iteration, the algorithm traverses the parallelotopes in $S$ and identifies those satisfying the selection condition $|A\cap P|\geq R|P|$, where $P$  denotes a parallelotope in $S$. Once such a parallelotope is found, it may include minterms not currently in 
$A$, so $A$  must be updated accordingly by removing the minterms now covered. The selected parallelotope is then synthesized in the circuit using the procedure in \algref{algorithm:1} (lines 7–9). This greedy selection strategy enables efficient coverage of minterms while reducing circuit cost by prioritizing higher-dimensional parallelotopes whenever possible.

\setlength{\textfloatsep}{5pt}
\begin{algorithm}[htbp]
\caption{SSHR-H}
\label{alg:H}
\begin{algorithmic}[1]
    \REQUIRE The input of minterms of the $f : \{0, 1\}^n \to \{0, 1\}$.
    \ENSURE Quantum circuit of Boolean function
    \STATE $circuit \gets INIT\_CIRCUIT(n+1)$
    \STATE $A\gets \text{on-set}$
    \STATE $\text{set}~S \gets \text{parallelotopes associated with}~A$
    \WHILE{\text{A is not empty}}
        \FOR{$P \in S $}
        \IF{$|P\cap A|/|P|\geq R$}
            \STATE $qc \gets Algorithm~\ref{algorithm:1}(P)$
            \STATE $\text{Add}~qc~\text{to the}~circuit$
            \STATE $\text{update}~A$
        \ENDIF
        \ENDFOR
    \ENDWHILE

    \RETURN $circuit$
\end{algorithmic}
\end{algorithm}

\subsection{SSHR-I}\label{sec:4b}

Given the inherent complexity and richness of the spatial structure in Boolean functions, which is often difficult to exhaustively characterize or analytically prove, we reformulate the Boolean  function synthesis task as an instance of the WP-SCP. 
This formulation allows us to explore the detailed synthesis capabilities of spatial structure in a more flexible manner, enabling the use of optimization tools such as integer linear programming (ILP) to guide circuit construction.

In essence, the goal is to ensure that each minterm in the on-set is covered an odd number of times, while each minterm in the off-set is covered an even number of times. This parity constraint guarantees that the final output of the synthesized quantum circuit correctly reflects the Boolean function. The process begins with the initialization of the quantum circuit, where no operations are applied and the output qubit is set to 0, corresponding to the quantum state $\ket{0}$. This serves as the baseline state from which all subsequent transformations are applied.

Formally, let $U = \text{on-set}, A = \text{All minterms}, A-U = \text{off-set}$, where the elements of each set correspond to the minterms of a given Boolean function. We define a collection of subsets \(S = \{p_1, p_2... ,p_n\}\), where each $p_i$ represents a parallelotope, i.e., a minterms subset of minterms derived from the spatial structure of the Boolean function. Each $p_i$ is associated with a weight that reflects the cost of implementing the corresponding quantum circuit block. The objective is to find a subset \(X\subset S\) that minimizes the total cost while satisfying the following parity constraints:
\begin{itemize}
    \item Each minterm in the on-set $U$ is covered an odd number of times by the subsets in $X$.
    \item Each minterm in the off-set $A-U$ is covered an even number of times.
\end{itemize}

This problem formulation corresponds to a Weighted Parity Set Covering Problem (WP-SCP), where the objective is to synthesize a cost-efficient quantum circuit that faithfully implements the intended Boolean function semantics. By transforming the Boolean function synthesis task into a WP-SCP instance, we enable the application of powerful integer linear programming (ILP) techniques, which can efficiently compute optimal or near-optimal circuit constructions. This modeling approach not only enhances synthesis accuracy but also offers greater flexibility in incorporating various optimization objectives and hardware constraints. It is evident that ESOP is a strict subset of our proposed method in terms of representation capability. Therefore, our method is guaranteed to produce a valid solution.

The notation and problem formulation adopted in this work are summarized in \tabref{tab:notation}. These definitions have been carefully aligned with the standard set covering problem framework to ensure consistency and clarity. This alignment facilitates the effective application of existing ILP solvers and optimization techniques to address the quantum circuit synthesis problem. The corresponding optimization model is defined as follows:

\begin{align}
    &V_j = \sum_{i\in S_i} x_i\cdot e_{ij} \quad \forall i \in S,\forall j \in U. \tag{C1}\\
    &V_k = 2\cdot y_k + 1 \quad \forall k \in A. \tag{C2}\\
    &V_l = 2\cdot z_l  \quad \forall l \in A-B. \tag{C3}\\
    &\forall i \in S. \tag{C4}\\
    &\forall j \in U. \tag{C5}\\
    &\forall k \in A. \tag{C6}\\
    &\forall l \in A-B. \tag{C7}
\end{align}

\begin{table}[tbp]
\centering
\caption{Weighted Parity SCP Notation}
\label{tab:notation}
\renewcommand{\arraystretch}{1.2}
\begin{tabular}{|ll|}
\hline
\multicolumn{2}{|c|}{Input Parameters}                                                               \\ \hline
\multicolumn{1}{|l|}{S}           & Set of parallelotopes                                                  \\
\multicolumn{1}{|l|}{A}           & Set of minterms of Boolean  function                               \\
\multicolumn{1}{|l|}{U}           & on-set             \\ \hline
\multicolumn{2}{|c|}{Indices}                                                                        \\ \hline
\multicolumn{1}{|l|}{i}           & Index of subsets in S                                            \\
\multicolumn{1}{|l|}{j}           & Index in U                                         \\
\multicolumn{1}{|l|}{k}           & Index of minterm in A                                         \\
\multicolumn{1}{|l|}{l}           & Index of inputs in $A-B$(the difference between set $A$ and $B$) \\ \hline
\multicolumn{2}{|c|}{Weighted Subsets}                                                               \\ \hline
\multicolumn{1}{|l|}{$C_i$}       & The CNOT cost of subset $S_i$                                    \\
\multicolumn{1}{|l|}{$G_i$}       & The T gate cost of subset $S_i$                      \\ \hline
\multicolumn{2}{|c|}{Binary Decision Variables}                                                      \\ \hline
\multicolumn{1}{|l|}{$x_i$}       & Whether $S_i$ is selected                                        \\
\multicolumn{1}{|l|}{$e_{ij}$} & Whether minterm $U_j \in S_i$                                      \\ \hline
\multicolumn{2}{|c|}{Integer Helper Variables}                                                       \\ \hline
\multicolumn{1}{|l|}{V}           & Cover times of minterm $v$                                         \\
\multicolumn{1}{|l|}{$y_k$}       & Cover the minterm in $A$ odd times                              \\
\multicolumn{1}{|l|}{$z_l$}       & Cover the minterm in $A-B$ even times                           \\ \hline
\multicolumn{2}{|c|}{Real Decision Variables}                                                        \\ \hline
\multicolumn{1}{|l|}{$TC$}    & Total cost of quantum circuit                                    \\ \hline
\end{tabular}
\end{table}

The final optimization goal is expressed in the objective function (OBJ), which seeks to minimize the total cost associated with the gates used in the quantum circuit. 

\begin{equation}
    \min(TC). \tag{OBJ}
\end{equation}

The first group of constraints, denoted as (C1), calculates the number of times each minterm is covered by the selected parallelotopes. Constraints (C2) and (C3) enforce the required parity conditions, ensuring that minterms in the on-set are covered an odd number of times, while those in the off-set are covered an even number of times. These constraints are essential for preserving the logical correctness of the synthesized quantum circuit. Constraints (C4) through (C7) impose bounds on the indices of subsets, minterms, and decision variables to ensure that all references remain within valid domains. Together, these constraints guarantee that the optimization model is structurally well-formed, semantically valid, and computationally tractable.

The introduction of the cost function $TC$, as defined in \eqnref{eq:cost_function}, provides greater flexibility in quantum circuit synthesis. By designing different objective functions—such as minimizing the number of CNOT gates, T gates, or other cost metrics—we can tailor the synthesis process to meet various optimization goals. The effectiveness of these cost-driven strategies is evaluated through experimental results presented in later sections.

\begin{equation}\label{eq:cost_function}
    TC = \alpha \sum_{i\in S}C_i \cdot x_i + \beta \sum_{i\in S}G_i \cdot x_i. 
\end{equation}

\subsection{Synthesis via SSHR-I}
In this section, we present the complete execution flow of the SSHR-I algorithm, as outlined in \algref{algorithm:2}. Lines 1–2 are responsible for initializing the inputs and outputs of the quantum circuit, while line 3 extracts the spatial structure associated with the on-set minterms. Line 4 formulates and solves the corresponding Weighted Set Cover Problem with Parity Constraints~(WP-SCP) using an ILP solver, which yields the selected set of parallelotopes. The selected parallelotopes are then passed through \algref{algorithm:1} in lines 5–9 to generate the corresponding quantum circuit blocks, which are subsequently appended to the overall circuit. The final output is the composed quantum circuit obtained by stitching together all circuit blocks.

Notably, since the order of the selected sets is arbitrary, the execution order of the corresponding circuit blocks does not affect the functional correctness of the final circuit. This property allows SSHR-I to naturally support layer-swappable circuit synthesis, enabling greater flexibility for subsequent optimization techniques, such as depth minimization or gate reordering.

\setlength{\textfloatsep}{5pt}
\begin{algorithm}[tbp]
    \caption{SSHR-I}
    \label{algorithm:2}
    \begin{algorithmic}[1]
        \REQUIRE The input of minterms of the \(f : \{0, 1\}^n \to \{0, 1\}\).
        \ENSURE Quantum circuit of Boolean  function.
        \STATE $circuit \gets INIT\_CIRCUIT(n+1)$
        \STATE $A \gets \text{on-set}$
        \STATE $\text{set}~S \gets \text{parallelotope associated with}~A$
        \STATE $\text{select\_set} = ILP\_SOLVER(A,S)$
        \FOR{$\text{parallelotope}~P$ in $\text{select\_set}$}
            \STATE $qc \gets Algoritm ~\ref{algorithm:1}(P)$
            \STATE $\text{Add}~qc~\text{to the}~circuit$
        \ENDFOR
        \RETURN $circuit$
    \end{algorithmic}
    \end{algorithm}

\subsection{An Example}
To illustrate the execution process of our proposed algorithm more intuitively, we present a concrete example. Specifically, we aim to synthesize a Boolean  function with $ID = 0x46B9$. Its corresponding representation in the hypercube is depicted in \figref{fig:example1}. This example will serve to demonstrate the step-by-step operation of the algorithm and highlight the role of spatial structure in the synthesis process.

We begin by computing the spatial structure embedded in the on-set $U$ of the Boolean function. This spatial structure, derived from the geometric representation of the function's minterms within the hypercube, forms a collection of parallelotopes denoted as the set $S$. For instance, we observe that the minterms [0011,0111,0000,0100] constitute a 2-dimensional parallelotope (i.e., $d=2$), as they satisfy the Minkowski sum described in \eqnref{eq:p}. 
Each parallelotope in $S$ corresponds to a potential circuit block and is associated with a different synthesis cost, which defines its weight in the WP-SCP formulation.
In accordance with the parity constraints, we must ensure that each minterm in the on-set $U$ is covered an odd number of times, while each minterm in the off-set $A-U$ is covered an even number of times. These parity conditions are essential to preserving the functional correctness of the synthesized quantum circuit, as they reflect the required constructive or destructive interference of quantum states.

Using an ILP solver, we obtain the optimal subset of parallelotopes that satisfies the parity constraints with minimal cost. In this example, the selected subsets are: $S_1 = [0000,0001,0010,0011,0100,0101,0110,0111], S_2 = [0110,1110], S_3 = [0001,0010,1001,1010]$. Each subset $S_i$ corresponds to a parallelotope extracted from the spatial structure of the Boolean function. According to our synthesis framework, each of these subsets can be implemented using a dedicated quantum circuit block following the procedure described in \algref{algorithm:1}. The resulting circuit blocks are shown in \figref{fig:example2}.

By sequentially combining these blocks, we construct the complete oracle for the Boolean function $f$. Notably, since the parity set covering formulation ensures that the order of the selected subsets does not affect the overall logic, the corresponding circuit blocks are mutually exchangeable. This modularity offers enhanced flexibility and paves the way for further circuit-level optimizations, such as gate reordering.

\begin{figure}[htbp]
\centering
\subfigure[Function with ID = 0x46B9.]{
\centering
\begin{minipage}[t]{0.4\linewidth}
\centering
\includegraphics[width=\linewidth]{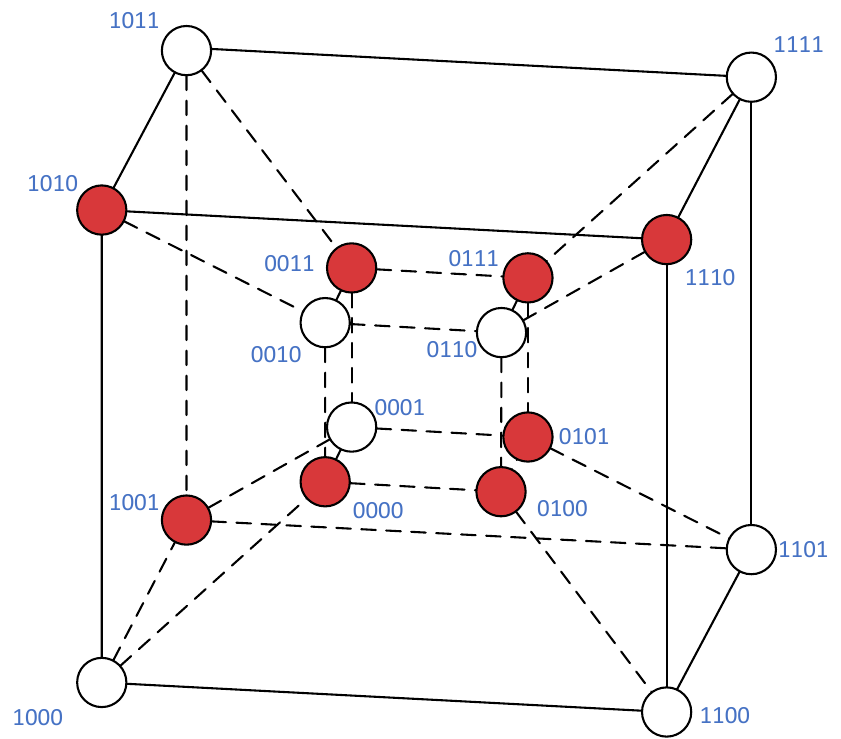}
\end{minipage}
\label{fig:example1}
}
\hspace{-5pt}
\subfigure[Corresponding quantum circuit.]{
\centering
\begin{minipage}[t]{0.5\linewidth}
\vspace{-80pt}
\label{fig:example2}
\hspace{-20pt}
\scalebox{1}{
\Qcircuit @C=1.0em @R=0.6em @!R { \\
        \nghost{{q}_{0} :  } & \lstick{{q}_{0} :  } & \qw & \qw & \qw \barrier[0em]{4} & \qw & \ctrlo{4} \barrier[0em]{4} & \qw & \qw & \qw \\
        \nghost{{q}_{1} :  } & \lstick{{q}_{1} :  } & \qw & \ctrlo{2} & \qw & \qw & \qw & \qw & \ctrl{1} & \qw \\
        \nghost{{q}_{2} :  } & \lstick{{q}_{2} :  } & \ctrl{1} & \qw & \ctrl{1} & \qw & \qw & \qw & \ctrl{1} & \qw \\
        \nghost{{q}_{3} :  } & \lstick{{q}_{3} :  } & \targ & \ctrl{1} & \targ & \qw & \qw & \qw & \ctrlo{1} & \qw \\
        \nghost{{q}_{4} :  } & \lstick{{q}_{4} :  } & \qw & \targ & \qw & \qw & \targ & \qw & \targ & \qw \\
\\ }}
\end{minipage}
}
\vspace{-1em}
\caption{An example of a quantum circuit for the oracle of a 4-bit Boolean function.}
\label{fig:example}
\end{figure}

\section{Evaluation}

\subsection{Environmental Setup}
We implement the synthesis program in Python, using Gurobi~\cite{gurobi} for the ILP solving. The experiments are conducted on a server with an Intel Xeon E5-2650 CPU and 512GB RAM.
Since the SSHR-I may be time-consuming, we have set a time limit of two minutes for each ILP run. 
If it exceeds 2 minutes, the process will be halted, and the best solution found so far will be output.
We test SSHR on all 3-bit and 4-bit Boolean functions, as well as 2000 randomly generated 5-bit and 6-bit Boolean functions.
To comprehensively validate SSHR, the number of minterms in the truth tables of these randomly generated n-bit Boolean functions is distributed between 1 and $2^{n-1}$.
We compared the synthesis results with an ESOP-based method~\cite{easy2024} and an XAG-based method~\cite{meuli2019role}.
During the synthesis process, to ensure a fair evaluation of the different algorithms, we calculate the cost of each $k$-MCT gate based on the decomposition costs outlined in \tabref{tab:decompose}.

\subsection{Results of SSHR-H}

\begin{table*}[bthp]
\centering
\vspace{-0.5em}
\caption{The number of gates used by the SSHR-H on Boolean  functions of different variables.}
\label{tab:greedy}
\begin{tabular}{|c|c|c|c|c|c|c|c|c|}
\hline
N & Algorithm    & X-gate & CNOT  & 2-MCT  & 3-MCT/Ancillary & 4-MCT & 5-MCT & 6-MCT \\ \hline
\multirow{3}{*}{3} & SSHR-H & 1100   & 560   & 220    & 128             & -     & -     & -     \\ \cline{2-9} 
  & ESOP   & 1061   & 172   & 333    & 136             & -     & -     & -     \\ \cline{2-9} 
  & XAG    & 1824   & 774   & 1190   & 579             & -     & -     & -     \\ \hline
\multirow{3}{*}{4} & SSHR-H & 2282   & 1094  & 249    & 218             & 90    & -     & -     \\ \cline{2-9} 
  & ESOP   & 2280   & 151   & 364    & 388             & 128   & -     & -     \\ \cline{2-9} 
  & XAG    & 4203   & 1620  & 2937   & 1466            & -     & -     & -     \\ \hline
\multirow{3}{*}{5} & SSHR-H & 40898  & 22913 & 1624   & 3913            & 2933  & 486   & -     \\ \cline{2-9} 
  & ESOP   & 43434  & 443   & 3054   & 6693            & 4861  & 1118  &  -    \\ \cline{2-9} 
  & XAG    & 77017  & 28338 & 64619  & 32292           & -     & -     & -     \\ \hline
\multirow{3}{*}{6} & SSHR-H & 89626  & 53684 & 915    & 3641            & 6798  & 3331  & 503   \\ \cline{2-9} 
  & ESOP   & 102740 & 303   & 2701   & 8543            & 11555 & 6997  & 1517  \\ \cline{2-9} 
  & XAG    & 151132 & 59232 & 144197 & 72085           & -     & -     & -     \\ \hline
\end{tabular}
\end{table*}
To control the selection criteria of candidate parallelotopes during circuit synthesis, we set the intersection threshold parameter $R=\frac{3}{4}$. That is, a parallelotope $P$ is selected only if the number of its minterms that intersect with the current working set $A$ over $\frac{3}{4}\cdot |P|$. This threshold balances two competing goals: maximizing the efficiency of spatial representation while avoiding unnecessary overlap and redundant computation. The rationale behind this choice lies in the scale of the target circuits. When $R$ is set too high (e.g., close to 1), a parallelotope must nearly be a subset of $A$ to be eligible for selection, which severely limits the algorithm's ability to exploit larger and higher-dimensional parallelotopes that may only partially overlap with $A$. As a result, potentially efficient structures would be ignored. Conversely, setting $R$ too low (e.g., below 0.5) results in the frequent selection of parallelotopes with overlap, leading to an increase in the number of iterations. This causes excessive reuse of overlapping spatial structures and generates redundant circuit components, reducing overall synthesis efficiency. Based on this trade-off, we empirically chose $R=\frac{3}{4}$ as a suitable compromise for small-scale Boolean  functions. The validity of the heuristic threshold is demonstrated in \tabref{tab:greedy}, where a dash (`-') indicates that certain gates do not appear for the given number of input variables and the column name: 3-MCT/Ancillary denotes Ancillary when the algorithm is XAG and 3-MCT for the rest. We evaluate the performance of the heuristic algorithm across various benchmark Boolean  functions with different numbers of input variables to assess its effectiveness.

We begin by analyzing the comparison between our proposed SSHR-H algorithm and  ESOP method. It is evident that SSHR-H results in significantly more CNOT gates than ESOP. This discrepancy can be attributed to two primary reasons. First, the ESOP algorithm directly targets the output qubit and applies CNOT gates only when the on-set entries in the truth table exhibit a high degree of correlation with specific input variables. In contrast, SSHR-H extensively employs CNOT gates to store and propagate intermediate results throughout the circuit, leading to a substantial increase in the total number of CNOT operations. However, this trade-off is compensated by SSHR-H's superior handling of MCT (multi-controlled Toffoli) gates. By leveraging intermediate computations stored via CNOTs, SSHR-H is able to reduce the number of control qubits required for MCT gates, which in turn leads to lower-cost circuit blocks. Since the cost of implementing MCT gates grows exponentially with the number of control qubits, reducing this number significantly improves overall synthesis efficiency. As a result, despite the increased number of CNOT gates, SSHR-H achieves better performance in terms of total quantum cost.

Next, we consider the comparison with XAG, which utilizes logic networks composed of $\{\oplus,\wedge,1\}$. XAG achieves its functionality by storing intermediate values in auxiliary qubits via Toffoli gates. This results in a higher demand for ancillary qubits compared to SSHR-H. Notably, SSHR-H is designed to avoid the use of ancillary qubits during synthesis—ancilla are introduced only when necessary during the decomposition of high-control MCT gates. This feature makes SSHR-H particularly suitable for near-term quantum architectures, where the availability of clean ancilla is often limited.

Furthermore, our heuristic algorithm is scalable and can be applied to Boolean  functions with a larger number of variables. To evaluate its performance, we conducted experiments for functions with $n=7,8$. The resulting total circuit costs were then compared against those produced by the XAG. A detailed summary of the comparative results is presented in \tabref{tab:789}, where a dash (`-') indicates that ESOP is unable to generate circuits for $n > 6$.

These results demonstrate that our algorithm maintains competitive performance as the function complexity increases, particularly in terms of reducing the number of CNOT gates and T gates while preserving resource efficiency.

\begin{table*}[]
\centering
\vspace{-0.5em}
\caption{Comparison of the number of gates used by SSHR-H and other algorithms.}
\label{tab:789}
\begin{tabular}{|c|c|c|c|c|c|c|c|c|c|c|c|}
\hline
 & \multicolumn{3}{c|}{\textbf{SSHR-H}}  & \multicolumn{4}{c|}{\textbf{ESOP}}    & \multicolumn{4}{c|}{\textbf{XAG}} \\ \hline
num\_vars & \multicolumn{1}{c|}{T-count}       & \multicolumn{1}{c|}{CNOT}   & Ancillary & \multicolumn{1}{c|}{T-count}      & \multicolumn{1}{c|}{CNOT}   & \multicolumn{1}{c|}{Ancillary} & CNOT gain & \multicolumn{1}{c|}{T-count}       & \multicolumn{1}{c|}{CNOT}    & \multicolumn{1}{c|}{Ancillary} & CNOT gain \\ \hline
3         & \multicolumn{1}{c|}{3588}    & \multicolumn{1}{c|}{3672}   & 128       & \multicolumn{1}{c|}{4507}   & \multicolumn{1}{c|}{4074}   & \multicolumn{1}{c|}{136}       & 9.86\%    & \multicolumn{1}{c|}{4760}    & \multicolumn{1}{c|}{7914}    & \multicolumn{1}{c|}{579}       & 53.60\%   \\ \hline
4         & \multicolumn{1}{c|}{7391}    & \multicolumn{1}{c|}{6540}   & 308       & \multicolumn{1}{c|}{11828}  & \multicolumn{1}{c|}{9047}   & \multicolumn{1}{c|}{516}       & 27.71\%   & \multicolumn{1}{c|}{11692}   & \multicolumn{1}{c|}{19148}   & \multicolumn{1}{c|}{1459}      & 65.85\%   \\ \hline
5         & \multicolumn{1}{c|}{159920}  & \multicolumn{1}{c|}{123573} & 7818      & \multicolumn{1}{c|}{280906} & \multicolumn{1}{c|}{176731} & \multicolumn{1}{c|}{13790}     & 30.08\%   & \multicolumn{1}{c|}{258420}  & \multicolumn{1}{c|}{415966}  & \multicolumn{1}{c|}{32285}     & 70.29\%   \\ \hline
6         & \multicolumn{1}{c|}{354525}  & \multicolumn{1}{c|}{233816} & 18107     & \multicolumn{1}{c|}{717499} & \multicolumn{1}{c|}{376925} & \multicolumn{1}{c|}{37126}     & 37.97\%   & \multicolumn{1}{c|}{576420}  & \multicolumn{1}{c|}{923824}  & \multicolumn{1}{c|}{72039}     & 74.69\%   \\ \hline
7         & \multicolumn{1}{c|}{541216}  & \multicolumn{1}{c|}{317756} & 29883     & \multicolumn{1}{c|}{-}      & \multicolumn{1}{c|}{-}      & \multicolumn{1}{c|}{-}         & -         & \multicolumn{1}{c|}{752032}  & \multicolumn{1}{c|}{1198414} & \multicolumn{1}{c|}{94004}     & 73.49\%   \\ \hline
8         & \multicolumn{1}{c|}{1081934} & \multicolumn{1}{c|}{632184} & 59053     & \multicolumn{1}{c|}{-}      & \multicolumn{1}{c|}{-}      & \multicolumn{1}{c|}{-}         & -         & \multicolumn{1}{c|}{1518272} & \multicolumn{1}{c|}{2419518} & \multicolumn{1}{c|}{189784}    & 73.87\%   \\ \hline
\end{tabular}
\end{table*}

\subsection{Results of SSHR-I}
\subsubsection{Results of CNOTs}
We first set the objective function to primarily optimize the number of CNOT gates and tested the performance on Boolean  functions with varying numbers of variables. \tabref{tab:t1} shows the number of T-counts, CNOT gates, and auxiliary qubits for different algorithms across variable counts ranging from 3 to 6, along with the percentage of CNOT reductions achieved by SSHR-I compared to other algorithms. As shown in the table, SSHR-I outperforms both the ESOP and XAG methods in terms of T-counts, CNOT gates, and auxiliary qubits in all test cases. Specifically, for 3, 4, 5, and 6 qubits, the average CNOT reductions are 21\%, 48\%, 56\%, and 54\%, respectively, compared to ESOP, and 59\%, 76\%, 81\%, and 81\%, respectively, compared to XAG for the same number of qubits. Additionally, it is evident that the reduction in the number of CNOT gates corresponds with a reduction in other gate types, demonstrating that our algorithm does not sacrifice the performance of other gates to achieve better circuit efficiency. As the number of variables $n$ increases, the gains become more significant; however, these improvements come at the cost of longer runtime for the ILP solver.

\begin{table*}[htbp]
\centering
\vspace{-0.5em}
\caption{The number of gates used by each algorithm when the objective of SSHR-I is set to minimize the number of CNOT gates.}
\label{tab:t1}
\begin{tabular}{|c|c|c|c|c|c|c|c|c|c|c|c|}
\hline
& \multicolumn{3}{c|}{\textbf{SSHR-I}} & \multicolumn{4}{c|}{\textbf{ESOP}} & \multicolumn{4}{c|}{\textbf{XAG}} \\ 
\hline
num\_vars & \multicolumn{1}{c|}{T-count} & \multicolumn{1}{c|}{CNOT}   & Ancillary & \multicolumn{1}{c|}{T-count} & \multicolumn{1}{c|}{CNOT}   & \multicolumn{1}{c|}{Ancillary} & CNOT gain & \multicolumn{1}{c|}{T-count} & \multicolumn{1}{c|}{CNOT}   & \multicolumn{1}{c|}{Ancillary} & CNOT gain \\ \hline
3         & \multicolumn{1}{c|}{3280}        & \multicolumn{1}{c|}{3232}   & 128       & \multicolumn{1}{c|}{4507}        & \multicolumn{1}{c|}{4074}   & \multicolumn{1}{c|}{136}       & 20.67\%   & \multicolumn{1}{c|}{4760}       & \multicolumn{1}{c|}{7914}   & \multicolumn{1}{c|}{579}       & 59.16\%   \\ \hline
4         & \multicolumn{1}{c|}{6028}       & \multicolumn{1}{c|}{4696}   & 212       & \multicolumn{1}{c|}{11828}       & \multicolumn{1}{c|}{9047}   & \multicolumn{1}{c|}{516}       & 48.09\%   & \multicolumn{1}{c|}{11692}       & \multicolumn{1}{c|}{19148}  & \multicolumn{1}{c|}{1459}      & 75.48\%   \\ \hline
5         & \multicolumn{1}{c|}{134656}      & \multicolumn{1}{c|}{78562}  & 5493      & \multicolumn{1}{c|}{280906}      & \multicolumn{1}{c|}{176731} & \multicolumn{1}{c|}{13790}     & 55.55\%   & \multicolumn{1}{c|}{258420}      & \multicolumn{1}{c|}{415966} & \multicolumn{1}{c|}{32285}     & 81.11\%   \\ \hline
6         & \multicolumn{1}{c|}{298267}      & \multicolumn{1}{c|}{171964} & 13191     & \multicolumn{1}{c|}{717499}     & \multicolumn{1}{c|}{376925} & \multicolumn{1}{c|}{37126}     & 54.38\%   & \multicolumn{1}{c|}{576420}     & \multicolumn{1}{c|}{923824} & \multicolumn{1}{c|}{72039}     & 81.39\%   \\ \hline
\end{tabular}
\end{table*}

\begin{table}[htbp]
\centering
\vspace{-0.5em}
\caption{Number of gates used when SSHR-I's target is T-count.}
\label{tab:t2}
\begin{tabular}{|c|ccc|}
\hline
\textbf{} & \textbf{}                    & \textbf{SSHR-I}               & \textbf{} \\ \hline
num\_vars & \multicolumn{1}{c|}{T-count} & \multicolumn{1}{c|}{CNOT}   & Ancillary \\ \hline
3         & \multicolumn{1}{c|}{2832}    & \multicolumn{1}{c|}{3579}   & 128       \\ \hline
4         & \multicolumn{1}{c|}{5742}    & \multicolumn{1}{c|}{5838}   & 208       \\ \hline
5         & \multicolumn{1}{c|}{110183}  & \multicolumn{1}{c|}{114320} & 5405      \\ \hline
6         & \multicolumn{1}{c|}{293765}  & \multicolumn{1}{c|}{262673} & 14179     \\ \hline
\end{tabular}
\end{table}

\subsubsection{Results of T gates}
As mentioned in section \secref{sec:4b}, we can flexibly set different optimization objectives. Since the other algorithms do not offer this capability, their experimental results are identical to those presented in \tabref{tab:t1}, and therefore, we will not repeat them. We optimized the T-count and tested the effect of varying numbers of variables.
\tabref{tab:t2} presents the number of T gates, CNOT gates, and auxiliary qubits for each method across the 3- to 6-variable test cases. As shown in the table, SSHR-I outperforms both the ESOP and XAG methods in terms of T-gates, CNOT gates, and ancillary qubits in all test cases. Specifically, the average T-gate reduction rates for the ESOP method are 37\%, 51\%, 61\%, and 59\% for 3, 4, 5, and 6 qubits, respectively, while the corresponding optimization rates for the XAG method are 41\%, 51\%, 57\%, and 49\% for the same qubit counts. Moreover, both CNOT and ancillary qubits also show significant reductions.

When comparing these results to those from experiments where the primary objective was CNOT minimization, we observe that the choice of objective function in SSHR influences the selection of the parallelotopes. In the T-count minimization case, there is a noticeable redundancy in CNOT gates compared to the CNOT-aim case, while the number of auxiliary qubits remains relatively stable. These results suggest that a flexible objective function can lead to targeted improvements across different metrics.
Notably, it can be seen that the number of CNOT gates increases compared to the results when our objective function is set to T-counts. 
This establishes a foundation for balancing the costs of different gates in circuit synthesis.

\subsection{Discussions}
In this section, we analyze the fundamental advantages of the SSHR strategy in terms of representation capability and circuit synthesis efficiency. Compared to the ESOP approach, SSHR demonstrates a significantly broader representation capacity. Specifically, the representation space of ESOP is a strict subset of that of SSHR, which allows SSHR to outperform ESOP in a variety of cases. While ESOP is limited to representing functions through disjoint product terms, SSHR leverages spatial structures as parallelotopes in the Boolean hypercube, enabling a more expressive and compact representation. A key strength of SSHR lies in its ability to utilize and reuse intermediate results during the synthesis process. This capability not only reduces redundant computations but also substantially enlarges the overall expressiveness of the circuit design space. Quantitatively, the representation capacity of ESOP for an $n$-variable Boolean  function can be estimated by the expression:\(\sum_{k=1}^n \binom{n}{k} 2^k\), which reflects the total number of distinct product terms. In contrast, the representation capacity of SSHR is determined by the number and dimensionality of extractable parallelotopes, the results are summarized and compared in \tabref{tab:expand}. As the number of variables increases, the gap between the two methods widens significantly.

We now analyze the underlying factors contributing to the advantages of SSHR over XAG. XAG, which employs heuristic strategies for logic synthesis, is particularly well-suited for large-scale circuits due to its focus on local optimizations and structural rewrites. However, this local perspective inherently limits its ability to exploit global patterns and structural structure embedded in the Boolean  function. In contrast, SSHR is designed with an emphasis on optimizing small-scale circuits, enabling it to leverage the global spatial structure of the Boolean  function, specifically the parallelotope structures described in \secref{spatial}. By analyzing the function's on-set as a whole and extracting higher-dimensional spatial regularities, SSHR achieves a deeper level of insight into the Boolean logic. This global awareness allows SSHR to make more informed decisions during synthesis, often resulting in more efficient quantum circuits in terms of gate count and overall cost.

\begin{table}[h]
\centering
\vspace{-0.5em}
\caption{Comparison of optimization space between SSHR and ESOP.}
\label{tab:expand}
\begin{tabular}{|c|c|c|c|}
\hline
\(n\) & ESOP & SSHR & Expansion factor \\ \hline
3 & 27                & 49         & 1.8           \\ \hline
4 & 81                & 257        & 3.2            \\ \hline
5 & 243               & 1539       & 6.3            \\ \hline
6 & 729               & 10299      & 14.1           \\ \hline
7 & 2187              & 75905      & 32.7          \\ \hline
8 & 6561              & 609441     & 92.9          \\ \hline
\end{tabular}
\end{table}
\section{Conclusion and Future Work}

In this paper, we proposed the Spatial Structure-based Hypercube Reduction~(SSHR), an effective approach for quantum Boolean  function synthesis.
We introduce two variants of SSHR: a heuristic algorithm, SSHR-H, and an ILP solver-based method, SSHR-I, to leverage this structure. 
SSHR enables the extraction of global information embeded in the parallelotopes of the Boolean  function and allows the storage of intermediate results without the need for auxiliary qubits.
Experimental results demonstrate that our method significantly outperforms existing techniques such as ESOP and XAG in key circuit metrics across all tested small-scale Boolean  functions. Looking forward, we aim to extend SSHR to support larger-scale quantum circuit synthesis, enhancing its scalability and generalizability. Additionally, we plan to explore its application to multi-output Boolean  function synthesis, further enriching its potential use cases in practical quantum computing systems.
\section*{Acknowledgments}
This research is supported by Innovation Program for Quantum Science and Technology (Grants No. 2021ZD0302901), National Natural Science Foundation of China (Grants No. 62002333) and  the Guangdong Provincial Quantum Science Strategic Initiative (Grants No.GDZX2403001).

\clearpage

\balance
\bibliographystyle{IEEEtran}
\bibliography{ref}

\begin{thebibliography}{10}
\providecommand{\url}[1]{#1}
\csname url@samestyle\endcsname
\providecommand{\newblock}{\relax}
\providecommand{\bibinfo}[2]{#2}
\providecommand{\BIBentrySTDinterwordspacing}{\spaceskip=0pt\relax}
\providecommand{\BIBentryALTinterwordstretchfactor}{4}
\providecommand{\BIBentryALTinterwordspacing}{\spaceskip=\fontdimen2\font plus
\BIBentryALTinterwordstretchfactor\fontdimen3\font minus \fontdimen4\font\relax}
\providecommand{\BIBforeignlanguage}[2]{{%
\expandafter\ifx\csname l@#1\endcsname\relax
\typeout{** WARNING: IEEEtran.bst: No hyphenation pattern has been}%
\typeout{** loaded for the language `#1'. Using the pattern for}%
\typeout{** the default language instead.}%
\else
\language=\csname l@#1\endcsname
\fi
#2}}
\providecommand{\BIBdecl}{\relax}
\BIBdecl

\bibitem{grover1996fast}
L.~K. Grover, ``A fast quantum mechanical algorithm for database search,'' in \emph{Proceedings of the twenty-eighth annual ACM symposium on Theory of computing}, 1996, pp. 212--219.

\bibitem{shor1999polynomial}
P.~W. Shor, ``Polynomial-time algorithms for prime factorization and discrete logarithms on a quantum computer,'' \emph{SIAM review}, vol.~41, no.~2, pp. 303--332, 1999.

\bibitem{harrow2009quantum}
A.~W. Harrow, A.~Hassidim, and S.~Lloyd, ``Quantum algorithm for linear systems of equations,'' \emph{Physical review letters}, vol. 103, no.~15, p. 150502, 2009.

\bibitem{simon1997power}
D.~R. Simon, ``On the power of quantum computation,'' \emph{SIAM journal on computing}, vol.~26, no.~5, pp. 1474--1483, 1997.

\bibitem{bravyi20226}
S.~Bravyi, J.~A. Latone, and D.~Maslov, ``6-qubit optimal clifford circuits,'' \emph{npj Quantum Information}, vol.~8, no.~1, p.~79, 2022.

\bibitem{montanaro2008quantum}
A.~Montanaro and T.~J. Osborne, ``Quantum boolean functions,'' \emph{arXiv preprint arXiv:0810.2435}, 2008.

\bibitem{fazel2007esop}
K.~Fazel, M.~A. Thornton, and J.~E. Rice, ``Esop-based toffoli gate cascade generation,'' in \emph{Proceedings of the 2007 IEEE Pacific Rim Conference on Communications, Computers and Signal Processing}, 2007, pp. 206--209.

\bibitem{meuli2022xor}
G.~Meuli, M.~Soeken, and G.~De~Micheli, ``Xor-and-inverter graphs for quantum compilation,'' \emph{npj Quantum Information}, vol.~8, no.~1, p.~7, 2022.

\bibitem{BernasconiBCCF23}
A.~Bernasconi, A.~Berti, V.~Ciriani, G.~M.~D. Corso, and I.~Fulginiti, ``{XOR-AND-XOR} logic forms for autosymmetric functions and applications to quantum computing,'' \emph{{IEEE} Trans. Comput. Aided Des. Integr. Circuits Syst.}, vol.~42, no.~6, pp. 1861--1872, 2023.

\bibitem{zi2024shallow}
W.~Zi, J.~Nie, and X.~Sun, ``Shallow quantum circuit implementation of symmetric functions with limited ancillary qubits,'' \emph{arXiv:2404.06052}, 2024.

\bibitem{patel2008optimal}
K.~N. Patel, I.~L. Markov, and J.~P. Hayes, ``Optimal synthesis of linear reversible circuits.'' \emph{Quantum Inf. Comput.}, vol.~8, no.~3, pp. 282--294, 2008.

\bibitem{chattopadhyay2014constructive}
A.~Chattopadhyay, S.~Majumder, C.~Chandak, and N.~Chowdhury, ``Constructive reversible logic synthesis for boolean functions with special properties,'' in \emph{Reversible Computation: 6th International Conference, RC 2014, Kyoto, Japan, July 10-11, 2014. Proceedings 6}.\hskip 1em plus 0.5em minus 0.4em\relax Springer, 2014, pp. 95--110.

\bibitem{fallek2016transport}
S.~Fallek, C.~Herold, B.~McMahon, K.~Maller, K.~Brown, and J.~Amini, ``Transport implementation of the bernstein--vazirani algorithm with ion qubits,'' \emph{New Journal of Physics}, vol.~18, no.~8, p. 083030, 2016.

\bibitem{debnath2016demonstration}
S.~Debnath, N.~M. Linke, C.~Figgatt, K.~A. Landsman, K.~Wright, and C.~Monroe, ``Demonstration of a small programmable quantum computer with atomic qubits,'' \emph{Nature}, vol. 536, no. 7614, pp. 63--66, 2016.

\bibitem{brayton1990multilevel}
R.~K. Brayton, G.~D. Hachtel, and A.~L. Sangiovanni-Vincentelli, ``Multilevel logic synthesis,'' \emph{Proceedings of the IEEE}, vol.~78, no.~2, pp. 264--300, 1990.

\bibitem{gover2010determinants}
E.~Gover and N.~Krikorian, ``Determinants and the volumes of parallelotopes and zonotopes,'' \emph{Linear Algebra and its Applications}, vol. 433, no.~1, pp. 28--40, 2010.

\bibitem{zhang2023fast}
Y.~Zhang, L.~Ni, J.~Zhang, G.~Luo, H.~Li, and S.~Zheng, ``Fast exact npn classification with influence-aided canonical form,'' in \emph{2023 IEEE/ACM International Conference on Computer Aided Design (ICCAD)}.\hskip 1em plus 0.5em minus 0.4em\relax IEEE, 2023, pp. 01--09.

\bibitem{zhang2022heuristic}
J.~Zhang, W.~Guo, G.~Yang, Y.~Zhu, and X.~Lv, ``A heuristic boolean npn equivalent matching verification method based on shannon decomposition,'' \emph{IEEE Access}, vol.~10, pp. 120\,369--120\,382, 2022.

\bibitem{zhang2023rethinking}
J.~Zhang, S.~Zheng, L.~Ni, H.~Li, and G.~Luo, ``Rethinking npn classification from face and point characteristics of boolean functions,'' in \emph{2023 Design, Automation \& Test in Europe Conference \& Exhibition (DATE)}.\hskip 1em plus 0.5em minus 0.4em\relax IEEE, 2023, pp. 1--6.

\bibitem{ambainis2014exact}
A.~Ambainis, J.~Gruska, and S.~Zheng, ``Exact quantum algorithms have advantage for almost all boolean functions,'' \emph{Quantum Information and Computation}, vol.~15, no. 5\&6, pp. 0435--0452, 2015.

\bibitem{nielsen2010quantum}
M.~A. Nielsen and I.~L. Chuang, \emph{Quantum computation and quantum information}.\hskip 1em plus 0.5em minus 0.4em\relax Cambridge university press, 2010.

\bibitem{miller2011elementary}
D.~M. Miller, R.~Wille, and Z.~Sasanian, ``Elementary quantum gate realizations for multiple-control toffoli gates,'' in \emph{2011 41st IEEE international symposium on multiple-valued logic}.\hskip 1em plus 0.5em minus 0.4em\relax IEEE, 2011, pp. 288--293.

\bibitem{divincenzo1995two}
D.~P. DiVincenzo, ``Two-bit gates are universal for quantum computation,'' \emph{Physical Review A}, vol.~51, no.~2, p. 1015, 1995.

\bibitem{barenco1995elementary}
A.~Barenco, C.~H. Bennett, R.~Cleve, D.~P. DiVincenzo, N.~Margolus, P.~Shor, T.~Sleator, J.~A. Smolin, and H.~Weinfurter, ``Elementary gates for quantum computation,'' \emph{Physical review A}, vol.~52, no.~5, p. 3457, 1995.

\bibitem{plesch2011quantum}
M.~Plesch and {\v{C}}.~Brukner, ``Quantum-state preparation with universal gate decompositions,'' \emph{Physical Review A}, vol.~83, no.~3, p. 032302, 2011.

\bibitem{gottesman1997stabilizer}
D.~Gottesman, \emph{Stabilizer codes and quantum error correction}.\hskip 1em plus 0.5em minus 0.4em\relax California Institute of Technology, 1997.

\bibitem{sikiric2014sum}
M.~D. Sikiri{\'c}, V.~Grishukhin, and A.~Magazinov, ``On the sum of a parallelotope and a zonotope,'' \emph{European Journal of Combinatorics}, vol.~42, pp. 49--73, 2014.

\bibitem{gurobi}
\BIBentryALTinterwordspacing
{Gurobi Optimization,Inc}. {G}urobi optimizer reference manual. 2024. [Online]. Available: \url{https://www.gurobi.com/}
\BIBentrySTDinterwordspacing

\bibitem{easy2024}
H.~Riener, ``easy,'' \url{https://github.com/hriener/easy}, 2024, version MIT-1-ov-file.

\bibitem{meuli2019role}
G.~Meuli, M.~Soeken, E.~Campbell, M.~Roetteler, and G.~De~Micheli, ``The role of multiplicative complexity in compiling low $ t $-count oracle circuits,'' in \emph{2019 IEEE/ACM International Conference on Computer-Aided Design (ICCAD)}.\hskip 1em plus 0.5em minus 0.4em\relax IEEE, 2019, pp. 1--8.

\end{thebibliography}



\end{document}